\documentclass[aps,pra,twocolumn,superscriptaddress]{revtex4-2}
\usepackage{graphicx}
\usepackage{amssymb}
\usepackage{amsmath}
\usepackage{array}
\usepackage{color}
\usepackage{hyperref}
\usepackage{physics}
\hypersetup{colorlinks=true,linkcolor=blue,citecolor=cyan}
\usepackage{graphicx}
\usepackage{cancel} 
\usepackage{bm}
\usepackage[english]{babel}
\usepackage{amsmath}
\usepackage{hyperref} 
\usepackage{times}
\usepackage{amsfonts,amssymb}
\usepackage{epsfig}
\usepackage{mathrsfs}
\usepackage{color,soul}
\usepackage{bbold}
\usepackage{textcomp}
\usepackage{dsfont}
\usepackage{braket}
\usepackage{empheq}
\usepackage{xcolor}
\usepackage{amsthm}
\usepackage{eurosym}
\usepackage{siunitx}
\AtBeginDocument{\RenewCommandCopy\qty\SI} 
\usepackage{enumitem}
\usepackage{comment}
\usepackage{multirow} 
\usepackage{orcidlink}

\def\tr{\operatorname{Tr}}
\newcommand{\pbraket}[3]{\left\langle #1 \middle| #2 \middle| #3 \right\rangle}

\def\Hi{\mathcal{H}}
\def\base{\{\ket{i}\}_{i=0}^{d-1}}
\def\BH{\mathcal{B}(\mathcal{H})}

\def\I{\mathcal{I}}
\def\IO{\mathcal{IO}}
\def\MIO{\mathcal{MIO}}
\def\SIO{\mathcal{SIO}}
\def\DIO{\mathcal{DIO}}

\newcommand{\kraus}[1]{\sum_k K_k #1 K_k^\dag}

\newcommand{\id}{\mathds{1}}

\newcommand{\ee}{\operatorname{e}}             
\newcommand{\ii}{\mathrm{i}}             

\renewcommand{\braket}[2]{\left\langle #1 \middle| #2 \right\rangle}

\usepackage{amsthm}
\newtheorem{definition}{Definition}
\newtheorem{lemma}{Lemma}

\newtheorem{proposition}{Proposition}
\newtheorem{example}{Example}

\begin{document}

\title{Resource theory of coherence in continuous position basis from measurement-induced dephasing}

\author{Karol Sajnok\orcidlink{0009-0004-5899-8923}}
\affiliation{Center for Theoretical Physics, Polish Academy of Sciences, Aleja Lotników 32/46, 02-668 Warsaw, Poland}
\affiliation{Nordita, Stockholm University and KTH Royal Institute of Technology, Hannes Alfvéns väg 12, 106 91 Stockholm, Sweden}
\affiliation{Institute of Theoretical Physics, University of Warsaw, ul. Pasteura 5, 02-093 Warsaw, Poland}
\email{ksajnok@cft.edu.pl}

\author{Fabio Costa\orcidlink{0000-0002-6547-6005}}
\affiliation{Nordita, Stockholm University and KTH Royal Institute of Technology, Hannes Alfvéns väg 12, 106 91 Stockholm, Sweden}
\affiliation{School of Mathematics and Physics, The University of Queensland, St Lucia, QLD 4072, Australia}
\email{fabio.costa@su.se}

\date{\today}

\begin{abstract}
We develop a resource-theoretic framework for quantum coherence directly in continuous basis, with emphasis on the position representation. Since position eigenstates are non-normalizable generalized eigenstates, the standard finite-dimensional dephasing map cannot be transferred directly to normal states. We therefore introduce a physically motivated dephasing channel based on random momentum kicks, equivalently described as the unconditional back-action of a finite-resolution position measurement. This yields a fixed-point notion of incoherence and a natural class of dephasing-covariant free operations. For physically relevant kernels, however, the fixed-point set contains no normal states, showing that continuous-basis coherence is tied to dephasing disturbance rather than to distance from a nonempty set of diagonal states. We study two quantifiers built from the channel action: a relative-entropy dephasing loss and a Hilbert--Schmidt dephasing loss. The former satisfies the main resource-theoretic properties under the free operations considered, while the latter is convex and experimentally transparent but fails monotonicity and strong monotonicity. We also formulate threshold witnesses for certifying coherence above a finite value and connect them, in a two-path setting, with interference visibility. Finally, we illustrate the framework with a Gaussian wavepacket evolving in a gravitational potential. The resulting theory provides a mathematically consistent and physically motivated treatment of coherence in continuous-variable systems.
\end{abstract}

\maketitle

\section{Introduction}
Quantum coherence is integral to quantum mechanics, underpinning phenomena such as interference and entanglement, and playing a central role in quantum technologies. In finite dimensions it is by now rigorously captured by resource theories that specify free states, free operations, and monotones with clear operational meanings. Yet many systems of practical and foundational interest---optical fields, motional degrees of freedom of atoms and nanoparticles, and gravitationally influenced matter waves---are intrinsically continuous-variable systems, where a direct transplant of the discrete framework runs into mathematical and physical pitfalls.

A basic issue already appears at the level of the reference basis itself. In the continuum, the natural candidates for ``incoherent'' basis vectors, such as position eigenstates $\ket{x}$, are not normalizable Hilbert-space vectors but generalized eigenstates, i.e. distributions. As a result, several constructions that are straightforward in finite dimensions become subtle in the continuum: projectors onto basis vectors are replaced by projection-valued measures, sums by integrals, and formal dephasing maps in the position basis need not define physically acceptable quantum channels on normal states. Another difficulty is definitional: a naive formal ``dephasing'' in the position basis is not trace preserving on normal states, and a physically implementable dephasing at finite resolution fails to be idempotent. A third difficulty is operational: one needs a class of free operations that neither create nor exploit coherence outside a prescribed spatial resolution while remaining compatible with completely positive dynamics. Finally, practical quantification requires measures that respect the resource-theoretic axioms and yet admit computable expressions or bounds for physically relevant families of states.

Here we show that the obstruction posed by the continuum is not merely technical, but structural. Starting from a finite-resolution position measurement, we derive a CPTP dephasing channel implemented by random momentum kicks and use it as the resource-destroying operation for position coherence. Its fixed-point structure is radically different from the finite-dimensional case: for physically relevant kernels, no normal trace-class state is exactly incoherent. This turns continuous-basis coherence into an intrinsically disturbance-based resource, quantified not by distance to a nonempty set of diagonal states, but by how strongly a state is changed by the dephasing channel. We identify the corresponding dephasing-covariant free operations, prove that the relative-entropy dephasing loss is monotone, strongly monotone, convex, and additive under the appropriate assumptions, and show that the natural Hilbert--Schmidt dephasing loss, although operationally accessible, fails monotonicity and strong monotonicity. We further explain why standard linear coherence witnesses become vacuous in this setting and replace them with threshold witnesses. Finally, we illustrate, with simulations, how the framework captures spatial coherence relative to a finite monitoring scale for a Gaussian wavepacket evolving in a gravitational potential.

The paper is structured as follows. Section~II reviews the standard resource theory of coherence in finite-dimensional systems, summarizing key concepts such as incoherent states, incoherent operations, and coherence measures. This sets the stage for our generalization to continuous basis. Section~III details and addresses the challenges in extending coherence to continuous variables, introduces a physically motivated dephasing map based on momentum kicks, and discusses the resulting incoherent states and coherence measures. We also explore a finite-resolution approach and introduce an idempotent step-function projector for defining incoherent states at fixed resolution. Measures of quantum coherence in continuous basis are discussed in Section~IV, where we adapt standard coherence quantifiers to the continuous setting and analyze their properties. Section~V introduces threshold witnesses for continuous-basis coherence and an double-slit interference inspired evaluation. In Section~VI, we illustrate our framework with the example of Gaussian states in a Newtonian potential, highlighting practical applications. Finally, Section~VII concludes with a summary of our findings and possible directions for future research.

\section{Resource Theory of Coherence: Finite–Dimensional Review}

The resource theory of coherence in finite-dimensional systems was formalized in \cite{baumgratz2014quantifying, levi2014quantitative, chitambar2016comparison, chitambar2016critical, chitambar2017erratum, winter2016operational, yadin2016quantum} and has become a standard framework for quantifying non-classicality. In this section we summarize its key ingredients—incoherent states, selected sets of incoherent operations, coherence measures—and recall their operational interpretations.

\subsection{Incoherent States}

Let $\Hi$ be a $d$-dimensional Hilbert space with a fixed, orthonormal reference basis $\base$. The \emph{incoherent states} \cite{streltsov2017colloquium} are defined as those density matrices $\rho \in \I \subset \BH$ which are diagonal in this specific basis,
\begin{equation}
\rho = \sum_{i=0}^{d-1} r_i\ketbra{i}{i},
\end{equation}
with $r_i\geq 0$, $\sum_i r_i=1$, where $\BH$ is the set of all linear operators on $\Hi$ and $\I$ is the set of incoherent states.

Analogously, this can be expressed in terms of a particular quantum channel, namely the dephasing (or completely decohering) channel in the chosen basis:
\begin{equation} \label{eq:dephasing_operator}
  \Delta(\rho) \coloneqq \sum_{i=0}^{d-1} \ketbra{i}{i}\rho\ketbra{i}{i}.
\end{equation}
Operationally, $\Delta$ can be interpreted as the non-selective projective measurement in the incoherent basis. Then, the set of incoherent states is defined as:
\begin{equation}
  \I = \left\{\delta\in\BH : \delta = \Delta(\delta) \right\}.
\end{equation}
Equivalently, $\I$ is the image of $\BH$ under the map $\Delta$. By construction, incoherent states contain no superposition between different basis vectors and serve as the \emph{free} or non-resourceful states of the theory.

\subsection{Incoherent Operations}

There are many classes of incoherent operations, as their purpose can vary depending on the context. Here, we recall four widely studied classes: maximally incoherent operations ($\MIO$), incoherent operations ($\IO$), strictly incoherent operations ($\SIO$) and dephasing-covariant incoherent operations ($\DIO$). For a broader discussion and additional classes, we refer the reader to \cite{streltsov2017colloquium}.

The $\MIO$ \cite{aberg2006quantifying} are the least restrictive class, consisting of all completely positive and trace-preserving (CPTP) maps $\Lambda_\MIO:\BH \to \BH$ that map incoherent states to incoherent states:
\begin{equation}
  \Lambda_\MIO (\I) \subseteq \I.
\end{equation}

Next, the most widely studied class of free operations is the set of $\IO$ \cite{baumgratz2014quantifying}, consisting of all CPTP maps $\Lambda_\IO:\BH\to\BH$ admitting a Kraus decomposition,
\begin{equation}
  \Lambda_\IO(\rho) = \kraus{\rho}, \quad \sum_k K_k^\dag K_k = \id,
\end{equation}
such that each $K_k$ maps incoherent states to incoherent states:
\begin{equation}
  \forall_k \, \forall_{\delta\in\I} \,K_k\delta K_k^\dag \in \I.
\end{equation}
Under $\IO$, no coherence can be generated from incoherent inputs, and any valid \emph{coherence measure} $C$ must be \emph{monotonic}:
\begin{equation}
  C\left(\Lambda_\IO(\rho)\right) \le C(\rho).
\end{equation}

Thirdly, the set $\SIO$ consists of those maps $\Lambda_\SIO \in \IO$ such that each Kraus operator in a Kraus decomposition is also strictly incoherent \cite{winter2016operational, yadin2016quantum}. Along with this property, outcomes of measurement in the reference basis $\base$ are independent of the coherence of the input state:
\begin{equation}
  \pbraket{i}{\Lambda_\SIO(\rho)}{i} = \pbraket{i}{\Lambda_\SIO(\Delta(\rho))}{i}, 
\end{equation}
where $\Delta \in \SIO$ is defined in Eq.~\eqref{eq:dephasing_operator}.

Finally, the set $\DIO$ is defined as the set of CPTP maps $\Lambda_\DIO:\BH\to\BH$ that commute with the dephasing operator \cite{chitambar2016comparison}:
\begin{equation}
  \Lambda_\DIO \circ \Delta = \Delta \circ \Lambda_\DIO.
\end{equation}
The set $\DIO$ is a strict subset of $\MIO$, whereas it is neither a subset nor a superset of $\IO$ or $\SIO$.

\subsection{Resourcefulness of coherence}
With $\IO$ defined, we can introduce a \emph{maximally coherent state}:
\begin{equation}
  \ket{\psi_d} = \frac1{\sqrt{d}}\sum_{i=0}^{d-1}\ket{i}.
\end{equation}
More generally, any state of the form $\frac1{\sqrt{d}}\sum_{i=0}^{d-1}\ee^{\ii \phi_i}\ket{i}$ is maximally coherent up to diagonal unitary phases. This state serves as a fundamental resource in the resource theory of coherence, as arbitrary states can be prepared from it under appropriately chosen $\IO$ operations \cite{baumgratz2014quantifying}. To quantify the coherence of a given state, two widely used methods are outlined below.

The \emph{distillable coherence} $C_{\rm dist}(\rho)$ represents the maximum rate \cite{yuan2015intrinsic, winter2016operational} at which copies of the maximally coherent qubit $\ket{\psi_2}$ can be distilled from $\rho$ using $\IO$. Conversely, the \emph{coherence cost} $C_{\rm cost}(\rho)$ is the minimum rate \cite{yuan2015intrinsic, winter2016operational} at which $\rho$ can be prepared from copies of $\ket{\psi_2}$ under $\IO$. These two quantities are related by the inequality:
\begin{equation}
  C_{\rm dist}(\rho) \le C_{\rm cost}(\rho).
\end{equation}
The explicit forms of distillable coherence and coherence cost can be found in \cite{streltsov2017colloquium}, while their relations to other coherence measures are discussed in the following subsection.

\subsection{Coherence Measures}

In this subsection, we recall coherence measures tailored to finite-dimensional systems, focusing on their mathematical properties and operational significance. These measures quantify the degree of non-classicality in a given quantum state relative to a fixed reference basis.

A functional $C:\BH\to[0,\infty)$ is a \emph{coherence measure} if it satisfies the following conditions \cite{baumgratz2014quantifying}:
\begin{enumerate}[label=(\roman*)]
  \item \emph{Faithfulness:} \label{c1}
  \begin{equation} 
    C(\rho) = 0 \Leftrightarrow \rho\in\mathcal{I}.
  \end{equation}
  \item \emph{Monotonicity:} \label{c2}
  \begin{equation}
    \forall_{\Lambda_\IO} \, C(\Lambda_\IO(\rho))\le C(\rho)
  \end{equation}
  \item \emph{Strong Monotonicity:} \label{c3}
  \begin{equation}
    C(\rho)\ge \sum \limits_k p_k\,C(\rho_k),
  \end{equation}
  where $p_k=\tr(K_k\rho K_k^\dagger)$, $\rho_k=\frac{K_k\rho K_k^\dagger}{p_k}$, for Kraus operators $\{K_k\}_k$ satisfying $\forall_k \, \forall_{\delta\in\I} \, K_k\delta K_k^\dagger\in\I$ and $\sum_k K_k^\dagger K_k=\id$.
  \item \emph{Convexity:} \label{c4}
  \begin{equation}
    C\left(\sum \limits_n q_n\rho_n\right)\le\sum \limits_n q_nC(\rho_n),
  \end{equation}
  where $q_n\geq 0$ and $\sum_n q_n = 1$.
\end{enumerate}
Additionally, the authors of \cite{streltsov2017colloquium} introduced two conditions analogous to those for entanglement measures \cite{plenio2014introduction}:
\begin{enumerate}[label=(\roman*), start=5]
  \item \emph{Uniqueness of pure states:} \label{c5}
  \begin{equation}
    C(\ketbra{\psi}{\psi}) = S(\Delta(\ketbra{\psi}{\psi})),
  \end{equation}
  where $\ket{\psi}$ is any pure state and $S(\rho)=-\tr(\rho\ln\rho)$ is the von Neumann entropy.
  \item \emph{Additivity:} \label{c6}
  \begin{equation}
    C(\rho \otimes \sigma) = C(\rho) + C(\sigma).
  \end{equation}
\end{enumerate}

The six conditions outlined above provide a rigorous framework for evaluating coherence quantifiers in quantum resource theories. They ensure that a measure is mathematically well defined, operationally meaningful, and consistent with the axioms of a resource theory. Yet, not all measures satisfy all six conditions simultaneously, and their applicability often depends on the specific operational setting. For example, the relative entropy of coherence (Eq.~\eqref{eq:rel_entropy_coherence} below), is particularly robust: it obeys all six conditions, including strong monotonicity and additivity, making it the preferred choice for tasks involving incoherent operations $\IO$. By contrast, elementwise $l_p$ norms of coherence are computationally simple and widely used in practice, but fail to meet some conditions such as monotonicity or additivity. Schatten $p$-norm based quantifiers provide further flexibility but generally fall short of satisfying the strongest requirements, especially under $\IO$. In the following we discuss these three main classes of measures in detail, emphasizing both their operational meaning and their compliance with the conditions \ref{c1}–\ref{c6}.

Before examining the individual measures, it is useful to recall that any distance-based quantifier of the form
\begin{equation}
    C_D(\rho) = \inf_{\sigma\in\mathcal I} D(\rho,\sigma),
\end{equation}
where $D$ is a contractive distance under CPTP maps, automatically satisfies faithfulness \ref{c1}, monotonicity under $\IO$ \ref{c2}, and convexity \ref{c4}. The relative entropy of coherence $C_\mathrm{rel}$ and the trace norm coherence $C_1$ are prime examples. In contrast, the Hilbert–Schmidt norm $C_2$ is not contractive, and therefore the corresponding quantifier fails strong monotonicity \ref{c3} \cite{baumgratz2014quantifying}.

\subsubsection{Relative entropy of coherence}
The relative entropy of coherence is defined as 
\begin{align} \label{eq:rel_entropy_coherence}
    C_{\mathrm{rel}}(\rho) &= S(\rho \Vert \Delta(\rho))=\tr(\rho\ln\rho) - \tr(\rho\ln\Delta(\rho)).
\end{align}
Due to the algebraic properties of the dephasing operator $\Delta$, namely $\tr(\rho\ln\Delta(\sigma))=\tr(\rho \Delta(\ln \Delta(\sigma)))=\tr(\Delta(\rho)\ln\Delta(\sigma))$, Eq.\eqref{eq:rel_entropy_coherence} can be expressed as
\begin{equation}
    C_{\mathrm{rel}}(\rho) = S(\Delta(\rho)) - S(\rho),
\end{equation}
where $S(\rho)=-\tr(\rho\ln\rho)$ is the von Neumann entropy and $\Delta(\rho)$ is the dephased state defined in Eq.\eqref{eq:dephasing_operator}.

This measure has a particularly transparent operational role. Under $\IO$, the distillable coherence equals $C_{\mathrm{rel}}$,
\begin{equation}
    C_\mathrm{dist}(\rho) = C_{\mathrm{rel}}(\rho),
\end{equation}
giving it the exact meaning of the optimal rate for distilling maximally coherent states \cite{winter2016operational,yuan2015intrinsic}. The dual task, coherence cost, is given by the entropy of formation $C_\mathrm{cost}(\rho)=E_\mathrm{form}(\rho)$ \cite{bennett1996mixed}. Inequality $C_\mathrm{cost}(\rho)\ge C_\mathrm{dist}(\rho)$ shows that coherence theory under $\IO$ is generally irreversible \cite{winter2016operational}. Reversibility is recovered under $\MIO$, where
\begin{equation}
    C_\mathrm{dist}(\rho)=C_\mathrm{cost}(\rho)=C_{\mathrm{rel}}(\rho),
\end{equation}
and no analogue of “bound coherence” exists \cite{brandao2015reversible}.

Importantly, $C_{\mathrm{rel}}$ satisfies all six criteria \ref{c1}-\ref{c6}, making it the most robust and widely accepted coherence monotone \cite{baumgratz2014quantifying,streltsov2017colloquium}. It also admits alternative interpretations: as the minimal noise needed to fully decohere a state \cite{singh2015maximally}, or as a measure of deviation from thermal equilibrium \cite{rodriguez2013thermodynamics}.

\subsubsection{Elementwise norms of coherence}

Another class of measures is based on elementwise $l_p$ norms of off-diagonal elements:
\begin{equation}
    C_{l_p}(\rho)=\Big(\sum_{i\ne j}|\rho_{ij}|^p\Big)^\frac{1}{p},
\end{equation}
with $C_{l_1}$ the canonical case. $C_{l_1}$ is faithful, convex, monotone under $\IO$, and has a simple closed form, making it the most widely used computable measure \cite{baumgratz2014quantifying}. However, it fails \ref{c5}, giving $C_{l_1}(\ket{\Psi_d})=d-1$ for maximally coherent states instead of $\log d$. Additionally, it also fails \ref{c6}, as $d_1d_2-1>(d_1-1)+(d_2-1)$, where $d_{1,2}$ are the dimensions of two Hilbert spaces. Finally, it also loses monotonicity \ref{c2} under $\MIO$ in dimensions $d>2$ \cite{bu2016note,streltsov2017colloquium}. Despite lacking an exact asymptotic operational meaning, its simplicity and link to interferometric visibility keep it valuable in practice.

For $p>1$, $C_{l_p}$ retains faithfulness \ref{c1} and convexity \ref{c4} but generally violates monotonicity \ref{c2} and strong monotonicity \ref{c3}, especially in higher dimensions \cite{rana2016trace}. These measures are therefore better viewed as coherence witnesses rather than full monotones.

\subsubsection{Schatten norm quantifiers}

A third family arises from Schatten $p$-norms,
\begin{equation}
    C_p(\rho)=\inf_{\sigma\in\mathcal I}|\rho-\sigma|_p,
\end{equation}
where $\|M\|_p = \left[\tr \left((M^\dagger M)^\frac{p}{2} \right) \right]^\frac{1}{p}$ is the Schatten $p$-norm \cite{schatten1970norm}.

The trace-norm case ($p=1$) is contractive and thus satisfies \ref{c1}, \ref{c2}, and \ref{c4}, but it fails strong monotonicity \ref{c3} under $\IO$ \cite{yu2016alternative}. For qubits and for two-qubit $X$ states, i.e.\ states whose density matrix has nonzero entries only on the main diagonal and anti-diagonal, $C_1$ coincides with $C_{l_1}$, restoring full monotonicity \cite{shao2015fidelity,rana2016trace}. For general states, identifying the closest incoherent state remains nontrivial \cite{chen2016quantifying}.

For $p=2$ (Hilbert–Schmidt) and higher, these measures preserve faithfulness and convexity but again violate monotonicity in most cases \cite{rana2016trace}. As such, they are not regarded as reliable coherence monotones, though they occasionally appear in applications due to computational convenience.

\subsubsection{Comparison of Properties}
The main properties of the measures discussed above are summarized in Table~\ref{tab:coherence-measures}. It is evident that the relative entropy of coherence stands out as the only quantifier satisfying all six conditions, while the others trade off formal rigor for computability and practical applicability.
\begin{center}
\begin{table}[h]
\begin{tabular}{lcccccc}
\hline\hline
Measure & \ref{c1} & \ref{c2} & \ref{c3} & \ref{c4} & \ref{c5} & \ref{c6} \\
\hline
$C_{\mathrm{rel}}$ & \checkmark & \checkmark & \checkmark & \checkmark & \checkmark & \checkmark \\
$C_{l_1}$ & \checkmark & \checkmark & \checkmark & \checkmark & \texttimes & \texttimes \\
$C_{l_p}$ ($p>1$) & \checkmark & \texttimes & \texttimes & \checkmark & \texttimes & \texttimes \\
$C_1$ & \checkmark & \checkmark & \texttimes & \checkmark & \texttimes & \texttimes \\
$C_p$ ($p>1$) & \checkmark & \texttimes & \texttimes & \checkmark & \texttimes & \texttimes \\
\hline\hline
\end{tabular}
\caption{Summary of coherence measures and their compliance with conditions \ref{c1}-\ref{c6} under $\IO$.}
\label{tab:coherence-measures}
\end{table}
\end{center}

\subsection{Coherence Witnesses and Experimental Realization}
In full analogy with entanglement, a coherence witness is a Hermitian operator $W$ such that
\begin{align}
    \forall_{\sigma \in \I}\; \tr(W\sigma) &\ge 0, \\
    \exists_{\rho \notin \I}\; \tr(W\rho) &< 0.
\end{align}
Witnesses arise naturally from the convex duality of distance- and robustness-based measures; in particular, the robustness of coherence admits a semidefinite-program formulation whose dual yields experimentally accessible witnesses \cite{napoli2016robustness,piani2016robustness}. Practical realizations include: (i) interferometric schemes that read out specific off-diagonal elements or visibilities, thereby implementing linear witnesses; (ii) phase-discrimination circuits that operationally certify nonclassical advantages linked to robustness; and (iii) quantum-Fisher-information-based lower bounds that require only a small set of measurements, providing witness-like certification of so-called ``unspeakable'' coherence, i.e.\ coherence relative to a physically distinguished observable or symmetry generator (such as energy, phase, or angular momentum), rather than relative to an arbitrary reference basis \cite{girolami2014observable}. These strategies have been demonstrated on photonic, NMR, and solid-state platforms, and integrate seamlessly with standard tomography when full state reconstruction is feasible \cite{streltsov2017colloquium}.

\subsection{Approaches to Coherence in Continuous Variables}

Most studies of coherence in continuous-variable systems have been developed in restricted frameworks rather than directly in a genuinely continuous basis such as position or momentum. This is not merely a technical matter: different approaches are motivated by different physical notions of classicality, and therefore identify different sets of free states and different kinds of superposition as relevant resources.

A first line of work takes the Fock basis $\{\ket{n}\}$ as the incoherent reference basis, leading to a countable, though infinite-dimensional, extension of the standard resource-theoretic framework. The main motivation here is clear: in many optical and bosonic settings, the occupation-number basis is singled out by the structure of the mode and by photon-counting measurements, so coherence is identified with superposition between different number sectors. Within this setting, standard measures such as the relative entropy of coherence remain well defined under finite-energy constraints, while the $l_1$ norm typically diverges and is therefore unsuitable \cite{zhang2016quantifying}. Thus, this framework preserves as much as possible of the finite-dimensional theory. However, it is clearly unsuitable to characterise non-classicality for continuous degrees of freedom, such as position.

A second approach focuses on Gaussian states, which are fully determined by first and second moments. The motivation here is mainly operational and experimental: Gaussian states are the natural workhorses of continuous-variable quantum information, because they are readily prepared in the laboratory, are stable under the most common optical operations, and admit a simple phase-space description in terms of covariance matrices and displacements \cite{adesso2014continuous}. One can then define Gaussian versions of relative-entropy or distance-based coherence quantifiers, but these usually provide only upper bounds to the corresponding unrestricted measures, since the closest incoherent state need not be Gaussian \cite{xu2016quantifying,buono2016quantum}. In this sense, the Gaussian framework is physically important and computationally tractable, but it is tailored to a privileged family of states rather than to coherence in a genuinely continuous basis.

A conceptually different line of work is the theory of optical coherence developed by Glauber and Sudarshan \cite{glauber1963quantum,sudarshan1963equivalence}. Here the underlying motivation is not basis-dependent superposition, but the distinction between classical and nonclassical radiation fields. In the Glauber--Sudarshan framework, coherent states $\ket{\alpha}$ are regarded as the basic classical pure states of the electromagnetic field, because they provide the closest quantum analogue of classical radiation: they reproduce classical field amplitudes at the level of expectation values, evolve within the family of coherent states under free dynamics, and underlie the semiclassical description of optical fields \cite{glauber1963quantum,sudarshan1963equivalence}. Correspondingly, a state is deemed classical when it can be written as a statistical mixture of coherent states with a nonnegative Glauber--Sudarshan $P$ representation, whereas nonclassicality appears when such a description requires a singular or nonpositive quasiprobability distribution \cite{glauber1963quantum,sudarshan1963equivalence,streltsov2017colloquium}. This notion of classicality is conceptually distinct from ``incoherence'' in the standard resource-theoretic sense: for instance, the coherent state
\begin{align}
    \ket{\alpha}=e^{-|\alpha|^2/2}\sum_{n=0}^\infty \frac{\alpha^n}{\sqrt{n!}}\ket n
\end{align}
is highly coherent in the Fock basis, even though it is classical in the Glauber--Sudarshan sense. The contrast between these approaches therefore reflects different physical motivations: one quantifies superposition relative to a chosen orthonormal basis, while the other quantifies departure from a classical-wave description of the field.

These examples make clear why coherence in the position or momentum basis requires separate treatment. Fock-space coherence is tied to number superselection and photon counting, Gaussian coherence to the privileged role of Gaussian states in continuous-variable protocols, and optical coherence to the distinction between classical and nonclassical radiation. None of these directly captures the question studied here: whether a state exhibits superposition between different positions or momenta, and how this superposition is degraded by coarse measurements or environmental noise. This is the natural setting, for instance, in finite-resolution position measurements, matter-wave interferometry, time-of-flight imaging, and scattering-induced decoherence, where the relevant loss of quantumness is the suppression of off-diagonal terms in the position or momentum representation.

\vspace{1em}
This review establishes the baseline definitions, measures, and operational insights that will guide our generalization to continuous‐basis systems in the next section.

\section{Generalizing Coherence to Continuous Variables}

The transition from a discrete, finite-dimensional resource theory of coherence to one based on observables with a continuous spectrum (such as position) is fraught with new technical and conceptual subtleties. We first catalogue the main obstacles, then develop a physically motivated ``momentum-kick'' dephasing map and its consequences (incoherent states, witnesses, and an illustrative Gaussian example). We also comment on a complementary, finite-resolution (Kraus) picture and, finally, introduce an idempotent step-function projector that cleanly defines incoherent states at fixed resolution in the continuous position basis.

Technical challenges arise from the infinite-dimensional nature of the Hilbert space and the continuous spectrum of observables. The main issues include:
\begin{enumerate}
  \item \textbf{Non-physical eigenstates and divergences.}
In a continuous Hilbert space $\mathcal H=L^2(\mathbb R)$ (the space of square-integrable wavefunctions on $\mathbb R$), the formal ``basis'' $\{\ket{x}\}_{x\in\mathbb R}$ of position eigenstates does not consist of valid Hilbert-space vectors and the corresponding ``projectors'' are not operators in Hilbert space. In particular, they are not trace-class, so that a naïve extension of the discrete dephasing map,
\begin{align}
  \widetilde\Delta(\rho)=\int_{-\infty}^{\infty}\dd{x}\ketbra{x}{x}\rho\ketbra{x}{x},
\end{align}
is ill-defined: using $\mel{x}{\rho}{y}=\rho(x,y)$ one obtains formally
\begin{align}
  \tr[\widetilde\Delta(\rho)] &=\int_{-\infty}^{\infty}\dd{x}\mel{x}{\rho}{x}\braket{x}{x} \nonumber\\
  &=\int_{-\infty}^{\infty}\dd{x}\rho(x,x)\,\delta(0),
\end{align}
so $\widetilde\Delta(\rho)$ fails to be a valid density operator.

\item \textbf{Breakdown of idempotency.}
In finite dimensions $\Delta^2=\Delta$. Any smooth, finite-resolution continuous-basis map $\mathcal D_\epsilon$ that approaches ideal dephasing as $\epsilon\to 0$ necessarily acts by suppressing (rather than deleting) off-diagonal matrix elements at finite $\epsilon$, and therefore typically satisfies $\mathcal D_\epsilon^2\neq \mathcal D_\epsilon$. The loss of projector structure complicates closure properties and standard monotonicity arguments.

\item \textbf{Trace-class and domain subtleties.}
Even when avoiding $\widetilde\Delta$, sharply structured kernels $\rho(x,y)$ can fall outside trace class or outside the domains of unbounded generators. Coherence functionals built from kernels therefore require a bona fide state-to-state map and explicit regularity assumptions.
\end{enumerate}

\subsection{Measurement-Induced Dephasing: Random Momentum Kicks as a CPTP Channel}

A natural way to motivate a position-basis dephasing channel is via an indirect measurement model \cite{busch1996quantum,busch1996standard,wiseman2009quantum,jacobs2006straightforward}. Consider a system $S$ and a meter $M$, both both continuous-variable, with canonical pairs $(\hat X_S,\hat P_S)$ and $(\hat X_M,\hat P_M)$. Let the meter be prepared in a wave packet $\ket{m}$ with position wavefunction $m(z)=\braket{z}{m}$, and the system in $\ket{\psi}$ with $\psi(x)=\braket{x}{\psi}$. Take the von Neumann measurement interaction
\begin{align} \label{eq:positionmeasurementH}
  \hat H=\lambda\,\hat X_S\otimes \hat P_M,
\end{align}
and evolve for time $t$:
\begin{align}
  \hat U\ket{\psi}\ket{m}
  &=\ee^{-\frac{\ii}{\hbar} \hat H t}\ket{\psi}\ket{m} \nonumber \\
  & =\int_{-\infty}^{\infty}\dd{x}\psi(x)\ket{x}\otimes \ee^{-\frac{\ii}{\hbar} \lambda t x \hat P_M}\ket{m}.
\end{align}
Since $\ee^{-\frac{\ii}{\hbar} a \hat P_M}$ is the translation operator on the meter position,
\begin{align}
  \bra{z}\ee^{-\frac{\ii}{\hbar} a \hat P_M}\ket{m}=m(z-a),
\end{align}
a measurement of the meter in the $\{\ket{z}\}$ basis yields the (unnormalized) conditional post-measurement system state
\begin{align}
  \left(\id\otimes\bra{z}\right)\hat U\ket{\psi}\ket{m} =\int_{-\infty}^{\infty}\dd{x}\psi(x)\,m(z-\lambda t x)\ket{x}.
\end{align}
The corresponding outcome density is
\begin{align}
  P(z)=\int_{-\infty}^{\infty}\dd{x}|\psi(x)|^2|m(z-\lambda t x)|^2,
\end{align}
which approaches an ideal position readout in the limit of a sharply peaked pointer state $m$.

For density operators, the (unnormalized) conditional update in the position representation takes the form
\begin{align}
  \rho(x,y) \mapsto \rho(x,y)\,m(z-\lambda t x)\,m^*(z-\lambda t y).
\end{align}
Averaging over outcomes $z$ (i.e., the unconditional evolution corresponding to discarding the meter readout) produces a position-dephasing channel that multiplies $\rho(x,y)$ by a function of $x-y$ set by the meter resolution. Explicitly, the unconditional map reads
\begin{align}
  \mathcal D_m(\rho) &\coloneqq \int_{-\infty}^{\infty}\dd{z} \hat M_z \rho \hat M_z^\dagger, \label{eq:meter_unconditional_line1} \\
  \hat M_z &\coloneqq \int_{-\infty}^{\infty}\dd{x} m(z-\lambda t x)\ketbra{x}{x}, \label{eq:meter_unconditional_line2}
\end{align}
so that, in the position representation,
\begin{align}\label{eq:meter_kernel}
  \big(\mathcal D_m(\rho)\big)(x,y) &= G\big(\lambda t(x-y)\big)\rho(x,y),
\end{align}
with the resolution-dependent suppression factor
\begin{align}\label{eq:G_def}
  G(z)\coloneqq \int_{-\infty}^{\infty}\dd{u} m(u) \, m^*(u+z).
\end{align}
By construction, $G(0)=\int \dd{z}\,|m(z)|^2=1$ and $|G(z)|\le 1$. More precisely, if the meter state $\ket{m}$ is normalizable, then $|G(z)|<1$ for every $z\neq 0$. Indeed, $G(z)=\bra{m}\hat T_z\ket{m}$, where $\hat T_z=\ee^{-\frac{\ii}{\hbar} z\hat P_M}$ is the translation operator on the meter. Since $\hat T_z$ is unitary, equality $|G(z)|=1$ would require $\ket{m}$ to be an eigenstate of $\hat T_z$. For $z\neq 0$, however, the corresponding generalized eigenstates are non-normalizable momentum eigenstates, so no normalizable wave packet can saturate the bound. We will use this strict inequality below when discussing the fixed-point structure of $\Delta_g$.

This construction makes explicit the operational meaning of ``dephasing'' in a continuous basis: it is the \emph{unconditional} state update associated with a measurement in (or close to) the incoherent basis, i.e., the channel obtained after averaging the conditional post-measurement states over all outcomes (equivalently, discarding the classical record). In the measurement-operator language, $\mathcal D_m(\rho)=\int \dd{y}\hat M_y \rho \hat M_y^\dagger$ is precisely the Kraus representation of a generalized position measurement with finite resolution determined by the pointer wavepacket $m$ \cite{busch1996quantum,busch1996standard,preskill2015lecture}. The resulting kernel action, Eq.~\eqref{eq:meter_kernel}, shows that finite-resolution position monitoring preserves the diagonal $\rho(x,x)$ while suppressing coherences $\rho(x,y)$ as a function of $x-y$, thereby implementing the physically intended notion of position-basis dephasing.

The dephasing from unconditional position measurements can also be written in terms of ``random momentum kicks''. Expanding the meter state in momentum space, the interaction generated by Eq.~\eqref{eq:positionmeasurementH} acts on the system, for each meter momentum component $p$, as the unitary $\ee^{-\frac{\ii}{\hbar}\lambda t\,p\,\hat X}$, i.e., as a momentum kick generated by $\hat X$. After tracing out the meter, one therefore obtains a convex mixture of such unitaries, with weights given by the meter momentum distribution \cite{wiseman2009quantum,jacobs2006straightforward}. Equivalently, the random-kick form provides the Fourier-dual representation of the same channel $\mathcal D_m$ whenever $\tilde g$ arises from the momentum-space distribution of the same pointer state that defines $G$ in Eq.~\eqref{eq:G_def}. More general choices of $\tilde g$ still define CPTP random-kick channels, but need not correspond to such a measurement model.

Accordingly, we model coarse position monitoring by imparting a random momentum kick $p$ with probability density $\tilde g(p)\ge 0$, $\int \dd{p}\,\tilde g(p)=1$, and (for unbiased kicks) $\int \dd{p}\,p\,\tilde g(p)=0$. Since $\hat X$ generates momentum translations, $\ee^{\frac{\ii}{\hbar}p \hat X}\ket{p_1}=\ket{p_1+p}$, we define the dephasing channel
\begin{align}\label{randomkicks}
  \Delta_g(\rho)=\int_{-\infty}^{\infty}\dd{p}\,\tilde g(p)\, \ee^{\frac{\ii}{\hbar}p \hat X}\,\rho\,\ee^{-\frac{\ii}{\hbar}p \hat X},
\end{align}
which is a convex mixture of unitaries and hence completely positive and trace preserving (CPTP). In fact, $\Delta_g$ is the same channel as $\mathcal D_m$, written in a representation adapted to translations generated by $\hat X$.

In the position representation, with $\rho(x,y)\coloneqq \mel{x}{\rho}{y}$,
\begin{align}
  \Delta_g(\rho) &=\iiint_{\mathbb{R}^3} \dd{p}\dd{x}\dd{y}\rho(x,y)\tilde g(p) \ee^{\frac{\ii}{\hbar}p \hat X} \ketbra{x}{y} \ee^{-\frac{\ii}{\hbar}p \hat X}\nonumber\\
  &=\iint_{\mathbb{R}^2} \dd{x}\dd{y} \left(\int_{\mathbb{R}} \dd{p}\ee^{\frac{\ii}{\hbar}p(x-y)}\tilde g(p)\right)\rho(x,y)\ketbra{x}{y}\nonumber\\
  &=\iint_{\mathbb{R}^2} \dd{x}\dd{y}g(x-y)\rho(x,y)\ketbra{x}{y},
\end{align}
where
\begin{align} \label{eq:FTg}
  g(\xi)\coloneqq \int_{-\infty}^{\infty}\dd{p}\ee^{\frac{\ii}{\hbar}p\xi}\tilde g(p)
\end{align}
is the characteristic function of $\tilde g$. Thus,
\begin{align}\label{positionaction}
  \Delta_g:\ \rho(x,y)\mapsto g(x-y)\rho(x,y).
\end{align}
Since $g(0)=\int \dd{p}\tilde g(p)=1$, the diagonal $\rho(x,x)$ is preserved, while off-diagonal elements are suppressed on a scale set by the width of $g$. Identifying $\Delta_g$ with the unconditional measurement channel $\mathcal D_m$, one has
\[
g(\xi)=G(\lambda t\,\xi),
\]
so the kernel in Eq.~\eqref{positionaction} is exactly the same suppression factor as in Eq.~\eqref{eq:meter_kernel}, written in the random-kick parametrization. This identification holds for the subclass of random-kick channels generated by the indirect-measurement construction above. In general, one may also consider CPTP channels of the form \eqref{randomkicks} with more singular probability measures $\tilde g$, for which the kernel $g$ need not arise from a normalizable pointer state.

From an operational point of view, we are mainly interested in dephasing maps that admit a measurement-model interpretation. In this setting, the function $g(\xi)$ quantifies the residual coherence between position components separated by $\xi$. A nontrivial measurement should suppress such coherences, so we typically have $|g(\xi)|<1$ for $\xi\neq 0$, while $g(0)=1$ guarantees that the diagonal part of the state is preserved. It is often physically natural, though not always necessary, to impose further conditions on the measurement model. For example, the absence of a net momentum kick corresponds to unbiased measurement noise, which entails appropriate symmetry properties of $g$, such as $g(\xi)=g(-\xi)^*$ and, in the common real-valued case, $g(\xi)=g(-\xi)$. These additional assumptions are not needed for most of the analysis and will be introduced explicitly only at the points where they play a role.

Unless $g$ takes only the values $0$ or $1$ a.e. (such as a top-hat or a Heaviside step function), $\Delta_g$ is not idempotent:
\begin{align}
  \Delta_g^2 : \rho(x,y)\mapsto g(x-y)^2\rho(x,y)\neq \Delta_g(\rho),
\end{align}
where $\Delta_g^2 \coloneqq \Delta_g\circ \Delta_g$ denotes the self-composition of $\Delta_g$, and $\circ$ indicates composition of maps. Nevertheless, $\Delta_g$ always has the identity as a fixed point,
\begin{align}\label{fixedidentity}
  \Delta_g(\id)=\id,
\end{align}
which follows immediately from Eq.~\eqref{positionaction} and $\mel{x}{\id}{y}=\delta(x-y)$, since $g(0)\delta(x-y)=\delta(x-y)$.


\subsection{Fixed-Point Incoherence in the Continuous Position Basis}\label{subsec:fixedpoint}

The measurement model above motivates $\Delta_g$ as the \emph{unconditional} (non-selective) evolution induced by position monitoring at finite resolution \cite{busch1996quantum, busch1996standard, jacobs2006straightforward, wiseman2009quantum}. In finite-dimensional coherence theory, incoherent states are fixed points of the dephasing projector. In the continuous-variable setting, the physically meaningful CPTP maps are not projectors in general, so one must specify what plays the role of the ``free'' set. In this work we adopt the fixed-point notion, which is both conceptually direct and mathematically sharp.

\begin{definition}[Fixed-point incoherent states]\label{def:fixedpoint_incoherent}
Fix a CPTP position-dephasing channel $\Delta_g$ as in Eqs.~\eqref{randomkicks}--\eqref{positionaction}. The set of \emph{$g$-incoherent} states is defined as the fixed-point set
\begin{align}\label{eq:Ifix_def}
  \mathcal I_{\mathrm{fix}}^g\coloneqq \big\{\rho \; \mathrm{state}:\ \Delta_g(\rho)=\rho\big\}.
\end{align}
\end{definition}

Using Eq.~\eqref{positionaction}, the fixed-point condition is equivalently
\begin{align}\label{eq:fixedpoint_kernel_condition}
  \bigl(1-g(x-y)\bigr)\rho(x,y)=0\qquad \text{for almost all }(x,y).
\end{align}
In particular, if $|g(\xi)|<1$ for all $\xi\neq 0$ (as for Gaussian resolution functions), then any fixed point must satisfy $\rho(x,y)=0$ for $x\neq y$. This reproduces the discrete-basis intuition that incoherent states are ``diagonal'' in the incoherent basis. The crucial difference is that, in a continuous basis, perfect diagonality is incompatible with normalizable density operators. We stress that the strict inequality $|g(\xi)|<1$ for $\xi\neq 0$ is an additional assumption on the kernel, not part of the definition itself. It is satisfied for the class of channels induced by the normalizable pointer-state measurement model discussed above, but need not hold for more general CPTP random-kick channels with singular probability measures $\tilde g$.

\begin{proposition}[Absence of normal fixed points]\label{prop:no_normal_fixed_points}
Assume $|g(\xi)|<1$ for all $\xi\neq 0$. Then $\mathcal I_{\mathrm{fix}}^g$ contains no normal (trace-class) density operators on $L^2(\mathbb R)$.
\end{proposition}

\begin{proof}[Proof sketch]
Let $\mathcal H=L^2(\mathbb R)$, and let $\mathcal T_1(\mathcal H)$ denote the space of trace-class operators on $\mathcal H$. Let $\rho\in\mathcal T_1(\mathcal H)$ and assume $\Delta_g(\rho)=\rho$. By Eq.~\eqref{positionaction},
\begin{align}
  (1-g(x-y))\,\rho(x,y)=0
\end{align}
in the sense of kernels. Since $1-g(\xi)\neq 0$ for $\xi\neq 0$, it follows that $\rho(x,y)=0$ for $x\neq y$, hence $\rho$ commutes with every multiplication operator $M_f$ ($f\in L^\infty(\mathbb R)$). Therefore $\rho$ belongs to the commutant of the von Neumann algebra $\{M_f:f\in L^\infty(\mathbb R)\}$. This algebra is maximal abelian (a MASA) in $B(L^2(\mathbb R))$, so its commutant is itself; thus $\rho=M_h$ for some $h\in L^\infty(\mathbb R)$ \cite{AnantharamanDelarocheGyeongju2012}.

On the non-atomic measure space $(\mathbb R,\mathrm{d}x)$, any nonzero multiplication operator $M_h$ is not compact; since every trace-class operator is compact, $M_h$ can be trace class only if $h=0$ a.e. \cite{reed2012methods}. Hence $\rho=0$, contradicting $\tr\rho=1$. Therefore no trace-class density operator can satisfy $\Delta_g(\rho)=\rho$.
\end{proof}

Proposition~\ref{prop:no_normal_fixed_points} has a clear physical reading: \emph{any} normalizable continuous-variable state has non-vanishing position coherence at some separation scale, and an exact dephasing projector (with nontrivial fixed points) cannot be realized by a regular CPTP channel at finite resolution. This makes fixed-point incoherence a stringent notion, aligned with the widely used intuition that physically preparable continuous-variable states generically possess nonzero coherence in any sharp continuous basis. 

For completeness, we note that one may also consider alternative notions of incoherence based on the image of the dephasing map rather than on its fixed points. In particular, one may call a state incoherent whenever it admits a physical preimage under $\Delta_g$, equivalently whenever the inverse kernel transform is again a state. Since this notion plays only an auxiliary role in the present work, we defer it to Appendix~\ref{app:alternative}.

A complementary option is to replace the smooth kernel $g(x-y)$ in Eq.~\eqref{positionaction} by a sharp top-hat mask,
\begin{align}
  g_\epsilon(x-y)=\Theta(\epsilon-|x-y|),
\end{align}
which leads to the step-function dephasing map
\begin{align}
  \Delta_\Theta^\epsilon:\ \rho(x,y)\mapsto \Theta(\epsilon-|x-y|)\rho(x,y).
\end{align}
Unlike the physically motivated CPTP channel $\Delta_g$, this map is exactly idempotent,
\begin{align}
  (\Delta_\Theta^\epsilon)^2=\Delta_\Theta^\epsilon,
\end{align}
and therefore defines a nontrivial finite-resolution free set through the fixed-point condition $\Delta_\Theta^\epsilon(\rho)=\rho$, equivalently $\rho(x,y)=0$ for $|x-y|>\epsilon$. The price is that the top-hat mask is not smooth and does not define a completely positive map. Thus it should be viewed as a mathematically sharp auxiliary construction rather than as a physical dephasing channel. We relegate its detailed discussion to Appendix~\ref{app:step}.

\section{Measures of Quantum Coherence in Continuous Variables}\label{sec:measures}

Given that $\mathcal I_{\mathrm{fix}}^g$ is empty within normal states under the assumptions of Proposition~\ref{prop:no_normal_fixed_points}, the role of a coherence measure is not to detect membership in a nontrivial free set, but to quantify \emph{the strength of coherence suppression} induced by the physically motivated dephasing $\Delta_g$.

The fixed-point viewpoint also implies that the standard finite-dimensional \emph{variational} definition,
\begin{align}\label{eq:Cvar_bad}
  C_{\mathrm{var}}(\rho)\coloneqq \inf_{\sigma\in\mathcal I_{\mathrm{fix}}^g} S(\rho\Vert\sigma),
\end{align}
is not well-posed on normal states. Since $\mathcal I_{\mathrm{fix}}^g$ contains no trace-class operators, the infimum would be taken over a set of non-normalizable distributions (formally $\sigma(x,y)=p(x)\delta(x-y)$), for which the relative entropy is undefined. We therefore treat Eq.~\eqref{eq:Cvar_bad} only as a statement of ill-definedness and adopt operational quantifiers built directly from the action of $\Delta_g$.

We consider two measures: (i) the relative-entropy $g$-dephasing loss, and (ii) a purity-type $l_2$ $g$-dephasing loss functional. Throughout, states $\rho$ are taken to be trace class, $\rho\in\mathcal T_1(\mathcal H)$, where $\mathcal H=L^2(\mathbb R)$. Since $\mathcal T_1(\mathcal H)\subset\mathcal T_2(\mathcal H)$, i.e.\ every trace-class operator is automatically Hilbert--Schmidt, one has in particular $\|\rho\|_2^2=\tr(\rho^2)<\infty$. We verify which of the axioms \ref{c1}--\ref{c6} each measure satisfies under the class of free operations $\DIO_g$ defined below.

\begin{definition}\label{def:Crel_g}
The \emph{relative-entropy $g$-dephasing loss} is
\begin{align}\label{eq:Crel_g}
  C_{\mathrm{rel}}^{g}(\rho)\coloneqq S\!\big(\rho\Vert\Delta_g(\rho)\big),
\end{align}
whenever the right-hand side is finite.
\end{definition}
In infinite dimensions, $S(\rho\Vert\Delta_g(\rho))$ need not be finite for arbitrary normal states. In particular, the relative entropy diverges whenever $\operatorname{supp}\rho\not\subseteq\operatorname{supp}\Delta_g(\rho)$. This can occur, for example, for sufficiently sharp momentum-kick models that effectively shift the support of the state so that the original state and its dephased version do not have compatible supports. Such a divergence reflects this support mismatch rather than simply a large amount of coherence. Accordingly, in the present work we restrict attention to states for which \eqref{eq:Crel_g} is well defined and finite.

If the dephasing function $g$ is generated by a physically regular measurement model, this pathology is typically absent. In particular, for unbiased measurement noise, corresponding to a vanishing mean momentum transfer, the map $\Delta_g$ does not induce a systematic displacement of the state in momentum space. Under the usual regularity assumptions on the noise distribution, the supports of $\rho$ and $\Delta_g(\rho)$ are then compatible for the class of states considered here, so that $C_{\mathrm{rel}}^{g}(\rho)$ remains finite. Thus, whenever we invoke the relative-entropy $g$-dephasing loss in such a measurement-model setting, we implicitly restrict to this physically regular regime.

\begin{definition}\label{def:C2_g}
The \emph{$l_2$ $g$-dephasing loss} is
\begin{align}\label{eq:C2loss_def}
  C_{2}^{g}(\rho)\coloneqq \|\rho\|_2^2-\|\Delta_g(\rho)\|_2^2 = \tr(\rho^2) - \tr\big((\Delta_g(\rho))^2\big).
\end{align}
\end{definition}

In what follows, most of the proofs concerning $C_2^g$ will be presented for dephasing maps that admit a measurement-model interpretation. In this setting, the associated function $g$ satisfies the contractivity condition $|g(\xi)|\leq 1$ for all separations $\xi$, together with $g(0)=1$. Under this assumption, $C_2^g$ is nonnegative and has a direct interpretation as the Hilbert--Schmidt loss induced by finite-resolution dephasing.

\begin{definition}\label{def:DIOg}
A CPTP map $\Lambda$ is a \emph{$g$-dephasing-covariant operation} if it commutes with $\Delta_g$:
\begin{align}\label{eq:DIOg}
  \Lambda\circ\Delta_g=\Delta_g\circ\Lambda.
\end{align}
We denote this class by $\DIO_g$.
\end{definition}

It is the natural class of operations under which the dephased state $\Delta_g(\rho)$ transforms consistently with the input state $\rho$.

\subsection{Faithfulness}\label{subsec:faithfulness}

\begin{lemma}[Nonnegativity and faithfulness of $C_{\mathrm{rel}}^{g}$]\label{lem:faith_Crel}
For any normal state $\rho$, $C_{\mathrm{rel}}^{g}(\rho)\ge 0$. Furthermore, if $|g(\xi)|<1$ for all $\xi\neq 0$, then the equation $C_{\mathrm{rel}}^{g}(\rho)=0$ has no solution among normal states.
\end{lemma}

\begin{proof}
Nonnegativity follows directly from Klein's inequality, which states that for any two density operators $\rho$ and $\sigma$, the relative entropy is non-negative $S(\rho\Vert\sigma) \ge 0$, with equality iff $\rho = \sigma$. In our case, the condition $C_{\mathrm{rel}}^{g}(\rho)=0$ implies that $\rho = \Delta_g(\rho)$, meaning $\rho$ must be a fixed point of the position-dephasing channel $\Delta_g$. Under the assumption that $|g(\xi)| < 1$ for all $\xi \neq 0$, it follows from Proposition~\ref{prop:no_normal_fixed_points} that no normal (trace-class) state can satisfy this fixed-point equation. Specifically, $\Delta_g$ strictly reduces the off-diagonal coherence in the position representation for any state that is not strictly ``diagonal'' in an unphysical (non-normal) sense. Thus, $C_{\mathrm{rel}}^{g}(\rho)$ is strictly positive for all normal states $\rho$.
\end{proof}

\begin{lemma}[Faithfulness properties of $C_2^{g}$]\label{lem:faith_C2}
The $l_2$ dephasing loss $C_2^{g}(\rho)=\|\rho\|_2^2-\|\Delta_g(\rho)\|_2^2$ is faithful with respect to the fixed-point condition:
\begin{align}
  C_2^{g}(\rho)=0 \quad \Longleftrightarrow \quad \Delta_g(\rho)=\rho.
\end{align}
Under the assumptions of Proposition~\ref{prop:no_normal_fixed_points}, this implies $C_2^{g}(\rho)>0$ for every normal state $\rho$.
\end{lemma}

\begin{proof}
For $C_2^{g}$, using Eq.~\eqref{eq:C2loss_def},
\begin{align} \label{eq:integral-form-c2g}
  C_2^{g}(\rho)=\iint_{\mathbb{R}^2} \dd{x}\,\dd{y}\,\bigl(1-|g(x-y)|^2\bigr)\,|\rho(x,y)|^2,
\end{align}
so $C_2^{g}(\rho)=0$ implies $(1-|g(x-y)|^2)|\rho(x,y)|^2=0$ a.e., hence $g(x-y)\rho(x,y)=\rho(x,y)$ a.e., i.e.\ $\Delta_g(\rho)=\rho$ in the Hilbert--Schmidt sense. The converse is immediate. Under Proposition~\ref{prop:no_normal_fixed_points}, no normal state satisfies $\Delta_g(\rho)=\rho$, hence $C_2^{g}(\rho)>0$ for all normal $\rho$.
\end{proof}

\subsection{Monotonicity under $\DIO_g$}\label{subsec:monotonicity}

\begin{lemma}[Monotonicity of $C_{\mathrm{rel}}^{g}$]\label{lem:mono_Crel}
For any $\Lambda\in\DIO_g$,
\begin{align}
C_{\mathrm{rel}}^{g}(\Lambda(\rho))\le C_{\mathrm{rel}}^{g}(\rho).
\end{align}
\end{lemma}

\begin{proof}
Since $\Lambda\in\DIO_g$, we have $\Lambda\circ\Delta_g=\Delta_g\circ\Lambda$. Therefore
$C_{\mathrm{rel}}^{g}(\Lambda(\rho)) = S\!\left(\Lambda(\rho)\,\middle\|\,\Delta_g(\Lambda(\rho))\right) = S\!\left(\Lambda(\rho)\,\middle\|\,\Lambda(\Delta_g(\rho))\right)$.
By the data-processing inequality for the quantum relative entropy, $S(\Lambda(\rho)\Vert \Lambda(\sigma))\le S(\rho\Vert \sigma)$, hence
$C_{\mathrm{rel}}^{g}(\Lambda(\rho)) \le S(\rho\Vert \Delta_g(\rho)) = C_{\mathrm{rel}}^{g}(\rho)$.
\end{proof}

\begin{proposition}
The functional $C_2^{g}$ is not monotone under $\DIO_g$.
\label{failure-mono}
\end{proposition}

\begin{proof}
Take a normalized $\phi\in L^2(\mathbb R)$ with a bounded support $C=\operatorname{supp}\phi$ and such that $C_2^g(|\phi\rangle\langle\phi|)>0$. Let $V_x$ and $P_S$ be the position-translation and projection operators, where $x\in \mathbb{R}$ and $S$ is an interval in $\mathbb{R}$. For $C=\operatorname{supp}\phi$, choose $a, b$ such that $A=C+a$ and $B=C+b$ are disjoint.

Define the CPTP channel $\Lambda(\rho) = K_A\rho K_A^\dagger + K_B\rho K_B^\dagger + K_0\rho K_0^\dagger$ using Kraus operators $K_A = V_{-a}P_A$, $K_B = V_{-b}P_B$, and $K_0 = P_{(A\cup B)^c}$. Since translations preserve relative distance, $\Lambda$ commutes with the dephasing map $\Delta_g$, ensuring $\Lambda \in \DIO_g$.

For $\phi_x = V_x\phi$, consider the input state
$\rho = \frac{1}{2}|\phi_a\rangle\langle\phi_a| + \frac{1}{2}|\phi_b\rangle\langle\phi_b|$. 
Due to disjoint supports and translation invariance of $C_2^g$, we have $C_2^g(\rho) = \frac{1}{2}C_2^g(|\phi\rangle\langle\phi|)$.

Applying $\Lambda$ yields $K_A\rho K_A^\dagger = K_B\rho K_B^\dagger = \frac{1}{2}|\phi\rangle\langle\phi|$ and $K_0\rho K_0^\dagger = 0$. Thus, $\Lambda(\rho) = |\phi\rangle\langle\phi|$, which leads to:
\begin{align}
    C_2^g(\Lambda(\rho)) = C_2^g(|\phi\rangle\langle\phi|) > \frac{1}{2}C_2^g(|\phi\rangle\langle\phi|) = C_2^g(\rho).
\end{align}
This strict increase violates monotonicity.
\end{proof}

\subsection{Strong monotonicity under $\DIO_g$}\label{subsec:strong_monotonicity}

\begin{lemma}[Strong monotonicity of $C_{\mathrm{rel}}^{g}$]\label{lem:strong_Crel}
Consider a quantum instrument $\{\Lambda_k\}_k$ with completely positive, dephasing-covariant branches; $\sum_k\Lambda_k$ is trace-preserving. Then
\begin{align}
    C_{\mathrm{rel}}^{g}(\rho)\ge \sum_k p_k\, C_{\mathrm{rel}}^{g}(\rho_k),
\end{align}
where $p_k=\tr[\Lambda_k(\rho)]$ and $\rho_k=\Lambda_k(\rho)/p_k$.
\end{lemma}

\begin{proof}
Define the flagged CPTP map $\mathcal M(\cdot)=\bigoplus_k \Lambda_k(\cdot)$. By the data-processing inequality,
\begin{align}
    S(\rho\Vert \Delta_g(\rho)) \ge S(\mathcal M(\rho)\Vert \mathcal M(\Delta_g(\rho))).
\end{align}
Since each branch is dephasing-covariant, $\Lambda_k\circ\Delta_g=\Delta_g\circ\Lambda_k$, and therefore
\begin{align}
    \mathcal M(\Delta_g(\rho)) = \bigoplus_k \Lambda_k(\Delta_g(\rho)) = \bigoplus_k \Delta_g(\Lambda_k(\rho)).
\end{align}
Using additivity of relative entropy on block-diagonal operators, together with $S(p\rho\Vert p\sigma)=p\,S(\rho\Vert \sigma)$, we obtain
\begin{align}
    S(\mathcal M(\rho)\Vert \mathcal M(\Delta_g(\rho)))
    &=\sum_k S(\Lambda_k(\rho)\Vert \Delta_g(\Lambda_k(\rho))) \nonumber\\
    &=\sum_k p_k\, S(\rho_k\Vert \Delta_g(\rho_k)) \nonumber\\
    &=\sum_k p_k\, C_{\mathrm{rel}}^{g}(\rho_k),
\end{align}
which proves the claim.
\end{proof}

\begin{proposition}
The functional $C_2^g$ fails strong monotonicity under $\DIO_g$.
\end{proposition}

\begin{proof}
Consider the instrument $\{\Lambda_A, \Lambda_B, \Lambda_0\}$ with branches $\Lambda_j(\rho) = K_j\rho K_j^\dagger$ and the same input state $\rho = \frac{1}{2}|\phi_a\rangle\langle\phi_a| + \frac{1}{2}|\phi_b\rangle\langle\phi_b|$ from the Proposition \ref{failure-mono}. We have already established that each branch belongs to $\DIO_g$ and that $C_2^g(\rho) = \frac{1}{2}C_2^g(|\phi\rangle\langle\phi|)$.

The branch probabilities are $p_A = p_B = \frac{1}{2}$ and $p_0 = 0$. The corresponding normalized post-measurement states are: $\rho_A = \frac{\Lambda_A(\rho)}{p_A} = |\phi\rangle\langle\phi|$ and $\rho_B = \frac{\Lambda_B(\rho)}{p_B} = |\phi\rangle\langle\phi|$. Evaluating the expected coherence after the measurement yields:
\begin{align}
    \sum_j p_j C_2^g(\rho_j) &= \frac{1}{2}C_2^g(\rho_A) + \frac{1}{2}C_2^g(\rho_B) = C_2^g(|\phi\rangle\langle\phi|) \nonumber \\
    &> \frac{1}{2}C_2^g(|\phi\rangle\langle\phi|) = C_2^g(\rho).
\end{align}
This strict increase demonstrates the failure of the strong monotonicity inequality.
\end{proof}

\subsection{Convexity}\label{subsec:convexity}

\begin{lemma}[Convexity]\label{lem:convexity_measures}
Both the relative-entropy dephasing loss $C_{\mathrm{rel}}^{g}$ and the $l_2$ $g$-dephasing loss $C_2^{g}(\rho)=\|\rho\|_2^2-\|\Delta_g(\rho)\|_2^2$ are convex.
\end{lemma}

\begin{proof}
Convexity of $C_{\mathrm{rel}}^{g}$ follows from joint convexity of $S(\cdot\Vert\cdot)$ and linearity of $\Delta_g$:
\begin{align}
  C_{\mathrm{rel}}^{g}\!\left(\sum_i q_i\rho_i\right)
  &= S\!\left(\sum_i q_i\rho_i \,\Big\Vert\, \Delta_g\!\left(\sum_i q_i\rho_i\right)\right) \nonumber\\
  &= S\!\left(\sum_i q_i\rho_i \,\Big\Vert\, \sum_i q_i\Delta_g(\rho_i)\right) \nonumber\\
  &\le \sum_i q_i S\!\big(\rho_i\Vert\Delta_g(\rho_i)\big).
\end{align}

For $C_2^{g}$, using Eq.~\eqref{eq:integral-form-c2g}, define the linear map $L_g$:  $\rho(x,y) \mapsto \sqrt{1-|g(x-y)|^2}\,\rho(x,y)$. Since $|g(\xi)|\le 1$ for a random-kick channel, this map is well defined on Hilbert--Schmidt kernels and $C_2^g(\rho)=\|L_g\rho\|_2^2$. The squared norm of a linear map is convex, hence $C_2^g$ is convex.
\end{proof}

\subsection{Uniqueness on pure states}\label{subsec:uniqueness_pure}
In discrete systems with idempotent dephasing $\Delta$, $S(\rho \Vert \Delta(\rho)) = S(\Delta(\rho))$ for pure states. In the continuous case, this fails because $\Delta_g$ is not a projection.

\begin{lemma}[Departure from entropy of the dephased state]\label{lem:c5_fail}
For a pure state $\rho=\ketbra{\psi}{\psi}$, the relative entropy of coherence reduces to 
\begin{align} \label{eq:Crel-pure}
    C_{\mathrm{rel}}^{g}(\rho) = -\bra{\psi}\ln(\Delta_g(\rho))\ket{\psi}.
\end{align}
For the pure state, the $l_2$ $g$-dephasing loss takes a simple form
\begin{align} \label{eq:C2-pure}
    C_2^{g}(\rho) = 1 - \tr(\Delta_g(\rho)^2).
\end{align}
\end{lemma}

\begin{proof}
For any pure state, $S(\rho) = 0$, so $S(\rho \Vert \sigma) = -\tr(\rho \ln \sigma) = -\bra{\psi}\ln \sigma \ket{\psi}$. Setting $\sigma = \Delta_g(\rho)$ yields the result for $C_{\mathrm{rel}}^{g}$. Similarly, since $\tr(\rho^2)=1$ for any pure state, $C_2^{g}$ becomes a simple measure of the ``purity loss'' induced by the channel $\Delta_g$. The identity $C_{\text{rel}} = S(\Delta_g(\rho))$ would require $\Delta_g$ to be a conditional expectation on pure states, which is precluded by its non-idempotency.
\end{proof}

\subsection{Additivity and Multiplicativity}\label{subsec:additivity}

\begin{lemma}[Additivity of $C_{\mathrm{rel}}^{g}$ and scaling of $C_2^g$]\label{lem:add_Crel}
For a product state $\rho_{AB}=\rho_A\otimes\rho_B$ and product dephasing $\Delta_{g_A} \otimes \Delta_{g_B}$, the relative entropy of coherence is additive 
\begin{align}
    C_{\mathrm{rel}}^{g_A,g_B}(\rho_{AB}) = C_{\mathrm{rel}}^{g_A}(\rho_A) + C_{\mathrm{rel}}^{g_B}(\rho_B).
\end{align}
In contrast, the $C_2^{g}$ obeys the relation 
\begin{align}
    C_2^{g_A,g_B}(\rho_{AB}) = &\tr(\rho_A^2)\tr(\rho_B^2) \nonumber \\
    &- \tr\big(\Delta_{g_A}(\rho_A)^2\big)\tr\big(\Delta_{g_B}(\rho_B)^2\big).
\end{align}
\end{lemma}

\begin{proof}
The additivity of $C_{\mathrm{rel}}^{g}$ follows directly from the additivity of quantum relative entropy on product states: $S(\rho_A \otimes \rho_B \Vert \sigma_A \otimes \sigma_B) = S(\rho_A \Vert \sigma_A) + S(\rho_B \Vert \sigma_B)$, applied to $\sigma_{A,B} = \Delta_{g_{A,B}}(\rho_{A,B})$. The identity $\tr[(A \otimes B)^2] = \tr(A^2)\tr(B^2)$ shows that $C_2^{g}$ does not scale linearly, unlike its entropic counterpart.
\end{proof}

\vspace{1em}

\paragraph*{Summary of the axioms.} 
The relative-entropy dephasing loss $C_{\mathrm{rel}}^{g}$ serves as the most robust quantifier within the fixed-point framework. Under $\DIO_g$ it satisfies non-negativity and fixed-point faithfulness \ref{c1}, monotonicity \ref{c2}, strong monotonicity \ref{c3}, convexity \ref{c4}, and additivity on product systems \ref{c6}, provided the relative entropy remains finite. Condition \ref{c5} does not hold in general because $\Delta_g$ is not an idempotent conditional expectation; instead, it reduces to the expectation value, Eq.~\eqref{eq:Crel-pure}.

In contrast, the $l_2$ $g$-dephasing loss $C_2^{g}$ is fixed-point faithful \ref{c1} in the Hilbert--Schmidt sense and convex \ref{c4}; however, according to Proposition~\ref{prop:no_normal_fixed_points}, it never vanishes on normal states. It fails monotonicity \ref{c2}, \ref{c3} under $\DIO_g$, and lacks additivity \ref{c6}, inheriting a non-linear scaling from the Hilbert--Schmidt norm. For pure states it reduces to Eq.~\eqref{eq:C2-pure}; in spite of remaining genuinely sensitive to coherence, it does not satisfy the standard entropic pure-state identity in \ref{c5}. 

The properties of these two measures are summarized in Table~\ref{tab:coherence_axioms}.

\begin{table}[h!]
\centering
\caption{Comparative analysis of coherence quantifiers against the resource-theoretic axioms \ref{c1}--\ref{c6} under $\DIO_g$.}
\label{tab:coherence_axioms}
\begin{tabular}{lcccccc}
\hline\hline
Measure & \ref{c1} & \ref{c2} & \ref{c3} & \ref{c4} & \ref{c5} & \ref{c6} \\ 
\hline
$C_{\mathrm{rel}}^{g}$ & \checkmark & \checkmark & \checkmark & \checkmark & \texttimes & \checkmark \\
$C_2^{g}$ & \checkmark & \texttimes & \texttimes & \checkmark & \texttimes & \texttimes \\
\hline\hline
\end{tabular}
\end{table}

\section{Threshold witnesses for continuous-basis coherence}\label{sec:witnesses}
In the fixed-point framework adopted in this work, the usual notion of a coherence witness relative to the set $\mathcal I_{\mathrm{fix}}^g$ becomes vacuous. Indeed, for physically relevant kernels Proposition~\ref{prop:no_normal_fixed_points} implies that $\mathcal I_{\mathrm{fix}}^g$ contains no normal states. Hence the condition
\begin{align}
  \Tr(W\sigma)\ge 0
  \qquad
  \forall\,\sigma\in\mathcal I_{\mathrm{fix}}^g
\end{align}
is empty on the physical state space. A meaningful replacement is therefore not a witness of exact incoherence, but a witness certifying that the amount of $g$-coherence exceeds a prescribed finite threshold.

Let $C^g$ be a convex coherence diagnostic, such as the relative-entropy dephasing loss $C_{\mathrm{rel}}^g$ or the Hilbert--Schmidt dephasing loss $C_2^g$.

\begin{definition}
    For a fixed threshold $c>0$, define the sublevel set
    \begin{align}\label{eq:threshold_set_general}
      \mathcal S_c^g
      \coloneqq
      \{\rho\ \text{state}: C^g(\rho)\le c\}.
    \end{align}
    Whenever $C^g$ is convex, $\mathcal S_c^g$ is a convex set.
\end{definition} 

\begin{definition}
    A threshold witness is a bounded Hermitian operator $W_c$ such that
    \begin{align}
      \forall\,\sigma\in\mathcal S_c^g \quad \Tr(W_c\sigma)\ge 0,
    \end{align}
    while for some state $\rho$,
    \begin{align}
      \Tr(W_c\rho)<0.
    \end{align}
\end{definition}
A negative expectation value certifies that $\rho\notin\mathcal S_c^g$, and therefore that
\begin{align}
  C^g(\rho)>c.
\end{align}
This restores the standard witness logic: the operator does not separate the state from an empty set of exactly incoherent states, but from the physically meaningful convex set of states whose $g$-coherence is below the chosen threshold.

For the relative-entropy measure, one may take $C^g=C_{\mathrm{rel}}^g$. Since $C_{\mathrm{rel}}^g$ is convex, the corresponding sublevel sets are convex and can be separated by bounded Hermitian operators. This gives threshold witnesses for relative-entropic $g$-coherence. For experimental purposes, however, a more transparent construction is obtained from the Hilbert--Schmidt dephasing loss.

For the random-kick channel, the $l_2$ dephasing loss can be written as
\begin{align}\label{eq:C2_integral_witness_section}
  C_2^g(\rho) = \iint \dd{x}\dd{y}\left(1-|g(x-y)|^2\right)|\rho(x,y)|^2.
\end{align}
Thus $C_2^g$ measures the Hilbert--Schmidt weight of the density matrix suppressed by the dephasing kernel. In particular, Eq.~\eqref{eq:C2_integral_witness_section} shows that $C_2^g$ is convex as a squared norm of a linear transformation of $\rho$, although, as shown above, it is not monotone under the full class $\DIO_g$. Therefore it can be used to define threshold witnesses, but not a full resource monotone.

\begin{example}
To make the construction explicit, we consider two normalized wavepackets $\ket L$ and $\ket R$, localized around positions separated by a distance $d$ and with negligible overlap. We assume that both packets are narrow compared with the scale on which $g$ varies. In this approximation, the dephasing map acts on the coherence between the packets as $\Delta_g(\ketbra{L}{R})\simeq g(d)\ketbra{L}{R}$ and $\Delta_g(\ketbra{R}{L})\simeq g(d)^*\ketbra{R}{L}$.

For a state supported on $\mathrm{span}\{\ket L,\ket R\}$, which we write as
\begin{align}
  \rho= \begin{pmatrix}
  p & c\\
  c^* & 1-p
  \end{pmatrix}_{\{L,R\}},
\end{align}
we obtain
\begin{align}\label{eq:C2_two_packet}
  C_2^g(\rho) \simeq 2\bigl(1-|g(d)|^2\bigr)|c|^2.
\end{align}
Thus, in this two-packet sector, $C_2^g$ is directly controlled by the off-diagonal amplitude $c$ and by the amount of coherence suppression at the separation $d$.

This gives an analytic threshold witness. We introduce the phase-sensitive interference observable
\begin{align}
  X_\theta \coloneqq e^{-\ii\theta}\ketbra{L}{R} + e^{\ii\theta}\ketbra{R}{L},
\end{align}
for which $\Tr(X_\theta\rho)=2\operatorname{Re}(e^{-\ii\theta}c)$. If a state $\sigma$ satisfies $C_2^g(\sigma)\le c_0$, then Eq.~\eqref{eq:C2_two_packet} implies
\begin{align}
  \left|\Tr(X_\theta\sigma)\right| \le \sqrt{\frac{2c_0}{1-|g(d)|^2}}.
\end{align}
The phase $\theta$ can be chosen, or scanned, so that the measured interference quadrature is positive and maximal. Therefore the operator
\begin{align}\label{eq:C2_threshold_witness}
  W_{c_0,\theta} \coloneqq \sqrt{\frac{2c_0}{1-|g(d)|^2}}\,\id - X_\theta
\end{align}
satisfies $\Tr(W_{c_0,\theta}\sigma)\ge 0$ for all states $\sigma$ with $C_2^g(\sigma)\le c_0$. Conversely, if an experiment gives
\begin{align} \label{eq:witness-condition-1}
  \left|\Tr(X_\theta\rho)\right| > \sqrt{\frac{2c_0}{1-|g(d)|^2}},
\end{align}
then $\Tr(W_{c_0,\theta}\rho)<0$, which certifies that $C_2^g(\rho)>c_0$.

For the balanced coherent superposition $\ket{\psi_\theta}=\frac{1}{\sqrt2}(\ket L+e^{\ii\theta}\ket R)$, we have $|c|=\frac{1}{2}$, and therefore
\begin{align}
  C_2^g(\ketbra{\psi_\theta}{\psi_\theta}) \simeq \frac{1-|g(d)|^2}{2}.
\end{align}
The maximum certifiable two-packet $g$-coherence therefore increases with the separation $d$ whenever $|g(d)|$ decreases. This reproduces the expected physical behavior: coherences between well-separated packets are more strongly affected by finite-resolution dephasing than coherences between nearby components.

The observable $X_\theta$ has the direct interpretation of the interference observable measured in a two-path experiment. In the double-slit setting, the localized states $\ket L$ and $\ket R$ represent the two paths, and detection at a point $x$ on the screen selects a relative propagation phase $\theta(x)$. In the far-field approximation, for slits placed symmetrically at positions $\pm \frac{d}{2}$, this phase is $\theta(x) \simeq \frac{kd}{L}x$, where $L$ is the distance to the screen and $k$ is the wavevector. The detected intensity is then, after normalizing the slowly varying single-slit envelope to unity,
\begin{align}
  I(x)=\frac12+\operatorname{Re}(e^{-\ii\theta(x)}c)
       =\frac12+\frac12\Tr(X_{\theta(x)}\rho).
\end{align}
Equivalently, with the present normalization of $X_\theta$, we have $\Tr(X_{\theta(x)}\rho)=2I(x)-1$. Thus $X_\theta$ measures precisely the oscillating part of the interference pattern, while its phase $\theta$ is scanned by moving along the detection screen. The witness condition, Eq.~\eqref{eq:witness-condition-1}, can therefore be read directly as a threshold condition on the observed fringe contrast:
\begin{align}
  I(x)>\frac{1}{2}\left(\sqrt{\frac{2c_0}{1-|g(d)|^2}}+1\right).
\end{align}
Equivalently, after maximizing over the phase, the usual visibility satisfies $V=\max_\theta\Tr(X_\theta\rho)=2|c|$, and hence
\begin{align}
  C_2^g(\rho)\simeq \frac{1-|g(d)|^2}{2}V^2.
\end{align}
In this sense, the threshold witness $W_{c_0,\theta}$ reduces, in the double-slit realization, to a visibility \cite{mandel1996optical, englert1996fringe} witness: sufficiently large interference fringes certify more than the threshold amount of $g$-coherence at separation $d$.
\end{example}

\section{Example: Gaussian Wavepacket in a Newtonian Potential}\label{sec:example}

As an illustration of the continuous-variable framework developed above, we consider a Gaussian wavepacket evolving in a Newtonian potential and evaluate the two coherence quantifiers introduced in Section~\ref{sec:measures}, namely the relative-entropy $g$-dephasing loss $C_{\mathrm{rel}}^{g}(\rho)=S(\rho\Vert\Delta_g(\rho))$ and the $l_2$ $g$-dephasing loss $C_2^{g}(\rho)=\|\rho\|_2^2-\|\Delta_g(\rho)\|_2^2$. This example is physically relevant for matter-wave interferometry and for massive particles subject to weak gravitational fields, where position coherence is naturally probed by coarse monitoring or random momentum-transfer models.

To obtain explicit formulas, we specialize to a Gaussian dephasing kernel. Let the momentum-kick distribution be
\begin{align}\label{eq:gaussian_kick_distribution}
  \tilde g(p)=\frac{1}{\sqrt{2\pi}\,\eta}\exp\!\left(-\frac{p^2}{2\eta^2}\right),
\end{align}
with momentum-spread parameter $\eta>0$. Its characteristic function is
\begin{align}\label{eq:gaussian_g_kernel}
  g(\xi)=\exp\!\left(-\frac{\eta^2\xi^2}{2\hbar^2}\right) =\exp\!\left(-\frac{\xi^2}{2\ell_g^2}\right),
\end{align}
where $\ell_g\coloneqq \frac{\hbar}{\eta}$ sets the spatial coherence scale of the monitoring channel. The action of $\Delta_g$ is therefore 
\begin{align}
    \Delta_g:\rho(x,y)\mapsto \exp(-\frac{(x-y)^2}{2\ell_g^2})\,\rho(x,y),
\end{align}
which is the explicit Gaussian dephasing kernel used throughout this example.

\subsection{Initial Gaussian state}

Consider a normalized Gaussian wavepacket with vanishing mean momentum,
\begin{align}\label{eq:initial_gaussian}
  \psi_0(x)=\frac{1}{(2\pi\sigma_0^2)^{1/4}} \exp\!\left[-\frac{(x-x_0)^2}{4\sigma_0^2}\right],
\end{align}
where $\sigma_0>0$ is the initial spatial width and $x_0\gg \sigma_0$. Let $\rho_0=\ketbra{\psi_0}{\psi_0}$. Using Eq.~\eqref{eq:C2-pure} together with Eqs.~\eqref{eq:initial_gaussian}-\eqref{eq:gaussian_g_kernel}, we find
\begin{align}\label{eq:C2_initial_gaussian}
  C_2^g(\rho_0) =1-\frac{1}{\sqrt{1+4\sigma_0^2/\ell_g^2}}.
\end{align}
In particular, $C_2^g(\rho_0)\to 0$ for $\frac{\sigma_0}{\ell_g}\to 0$, while it approaches $1$ for $\frac{\sigma_0}{\ell_g}\to\infty$.

For the relative-entropy dephasing loss, although the exact closed form is less transparent than for $C_2^g$, Jensen's inequality applied to Eq.~\eqref{eq:Crel-pure} yields the lower bound
\begin{align}\label{eq:Crel_initial_bound}
  C_{\mathrm{rel}}^g(\rho_0) &\ge -\ln\!\Big(\bra{\psi_0}\Delta_g(\rho_0)\ket{\psi_0}\Big) = \frac12 \ln\!\left(1+\frac{2\sigma_0^2}{\ell_g^2}\right).
\end{align}
Thus both coherence quantifiers increase monotonically with the ratio $\frac{\sigma_0}{\ell_g}$: broader wavepackets possess larger position coherence relative to the fixed monitoring scale $\ell_g$.

\subsection{Evolution in a Newtonian potential}

Let a particle of mass $m$ evolve on $x>0$ under the Hamiltonian
\begin{align}
  \hat H=\frac{\hat p^2}{2m}-\frac{GmM}{\hat x}.
\end{align}
For times during which the packet remains well localized around $x_0$ with $x_0\gg \sigma_0$, we may expand the potential to second order as $V(x)\approx V(x_0)+V'(x_0)(x-x_0)+\frac12V''(x_0)(x-x_0)^2$, where $V(x)=-\frac{GmM}{x}$, $V'(x_0)=\frac{GmM}{x_0^2}$, and $V''(x_0)=-\frac{2GmM}{x_0^3}$. The center follows the classical trajectory $x_c(t)\approx x_0+\frac12 a_0 t^2$ with $a_0=-\frac{GM}{x_0^2}$, while the curvature corresponds to an inverted harmonic oscillator with $\Omega^2=V''(x_0)/m=-\frac{2GM}{x_0^3}$.

Defining $\kappa\coloneqq \sqrt{\frac{2GM}{x_0^3}}$, the Gaussian packet $\psi_t$ remains Gaussian within this quadratic approximation, with time-dependent width
\begin{align}\label{eq:sigma_t_newton}
  \sigma_t^2 \approx \sigma_0^2\cosh^2(\kappa t) + \left(\frac{\hbar}{2m\sigma_0\kappa}\right)^2\sinh^2(\kappa t).
\end{align}

Since the dephasing kernel depends only on the separation $x-y$, both $C_2^g$ and the lower bound for $C_{\mathrm{rel}}^g$ depend only on the instantaneous width $\sigma_t$, not on the packet center $x_c(t)$. Replacing $\sigma_0$ by $\sigma_t$ in Eqs.~\eqref{eq:C2_initial_gaussian}-\eqref{eq:Crel_initial_bound} therefore gives
\begin{align}\label{eq:C2_time_gaussian}
  C_2^g(\ketbra{\psi_t}{\psi_t}) &= 1-\frac{1}{\sqrt{1+4\sigma_t^2/\ell_g^2}}, \\
  C_{\mathrm{rel}}^g(\ketbra{\psi_t}{\psi_t}) &\ge \frac12 \ln\!\left(1+\frac{2\sigma_t^2}{\ell_g^2}\right).\label{eq:Crel_time_bound}
\end{align}

Eqs.~\eqref{eq:C2_time_gaussian}-\eqref{eq:Crel_time_bound} are strictly increasing functions of $\sigma_t$. Thus, within the present framework, the coherence of the evolving packet increases as the state becomes more spatially delocalized relative to the fixed monitoring scale $\ell_g$. In the Newtonian case, the inverted curvature enhances spreading compared with free motion, so the coherence quantified by $C_2^g$ and by the lower bound on $C_{\mathrm{rel}}^g$ grows correspondingly in time within the validity of the quadratic approximation.

As an additional illustration, in Fig.~\ref{fig:grav-decoh} we plot the lower-bound expression for the relative-entropy dephasing loss relative to its initial value. The parameter choice is inspired by the scales commonly discussed in Bose--Marletto--Vedral (BMV) and related quantum gravity-induced entanglement of matter proposals: mesoscopic masses of order $10^{-14}\,\mathrm{kg}$, spatial separations or closest approaches of order $10^{-4}\,\mathrm{m}$, and interaction times of order seconds \cite{marletto2017gravitationally, schut2024micrometer, aziz2025classical}.

\begin{figure}[t]
    \centering
    \includegraphics[width=\linewidth]{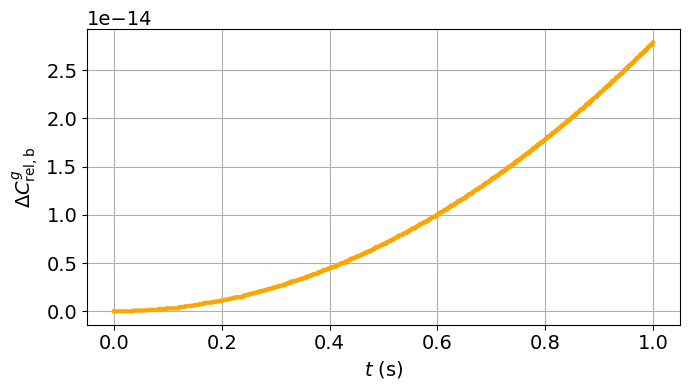}
    \caption{Time dependence of the lower-bound expression for the relative-entropy dephasing loss, shown relative to its initial value for a Gaussian wavepacket evolving in a gravitational potential. We plot $\Delta C_{\mathrm{rel,b}}^g(t)=C_{\mathrm{rel,b}}^g(\rho_t)-C_{\mathrm{rel,b}}^g(\rho_0)$, where $C_{\mathrm{rel,b}}^g$ denotes the lower-bound expressions given by Eqs.~\eqref{eq:Crel_initial_bound} and \eqref{eq:Crel_time_bound}. The parameters are $m=M=10^{-14}\,\mathrm{kg}$, $x_0=200\,\mu\mathrm{m}$, $\sigma_0=10\,\mu\mathrm{m}$, $\ell_g=20\,\mu\mathrm{m}$, and $t_{\max}=1\,\mathrm{s}$.}
    \label{fig:grav-decoh}
\end{figure}

The evolution of a Gaussian wavepacket in a gravitational potential illustrates the central feature of the present continuous-variable construction: coherence is measured relative to a physically motivated spatial-resolution scale encoded in the dephasing channel $\Delta_g$. Narrow packets with $\sigma_t\ll \ell_g$ are only weakly affected by the monitoring and therefore carry little position coherence in this sense, whereas broad packets with $\sigma_t\gtrsim \ell_g$ possess substantial coherence across spatial separations that are resolved by the dephasing map. The Newtonian evolution then provides a simple dynamical setting in which this physically defined coherence changes in a controlled and analytically tractable manner.

\section{Conclusion and Outlook}
We have developed a resource-theoretic framework for coherence in the continuous position basis based on a physically motivated dephasing channel. In contrast to finite dimensions, the natural reference vectors in the continuum are generalized eigenstates rather than normalizable states, so sharp basis dephasing cannot be directly used as a physical resource-destroying map. We instead define operationally meaningful dephasing through a CPTP channel generated by random momentum kicks, equivalently by finite-resolution position monitoring.

This approach leads to a stringent fixed-point notion of incoherence: for nontrivial dephasing kernels, no normal trace-class state is exactly incoherent. Within this framework, we identified the natural class of free operations as those commuting with the dephasing map and studied two coherence quantifiers. The relative-entropy dephasing loss turns out to be the more robust measure, satisfying the main resource-theoretic requirements under the free operations considered. By contrast, the $l_2$ dephasing loss, although convex and experimentally transparent, fails key monotonicity properties and therefore should be viewed as a more limited quantifier.

We also showed how coherence can be certified operationally through threshold witnesses. In the two-path case, these witnesses acquire a direct interferometric meaning: sufficiently large fringe visibility certifies that the state contains more than a prescribed amount of $g$-coherence at the path separation. Finally, the example of a Gaussian wavepacket in a gravitational potential illustrates how the framework captures spatial coherence relative to a finite monitoring scale in a concrete dynamical scenario relevant to BMV-inspired regimes.

Natural directions for further work include applications to non-Markovian, relativistic, and field-theoretic settings.

\section*{Acknowledgments}
We acknowledge support from the International Network on Acausal Quantum Technologies (INAQT), supported by the EPSRC under grant EP/W026910/1, and from the Knut and Alice Wallenberg Foundation through the Wallenberg Initiative for Network and Quantum Information (WINQ). The authors have benefited from the activities of COST Action CA23115: Relativistic Quantum Information, funded by European Cooperation in Science and Technology. KS also thanks Nordita for providing a summer internship opportunity.

\bibliography{bibliography}


\appendix

\section{Alternative notions of incoherence}\label{app:alternative}

The main text adopts the fixed-point notion of incoherence associated with the CPTP dephasing channel $\Delta_g$ in Eqs.~\eqref{randomkicks}--\eqref{positionaction}. For physically relevant kernels satisfying $|g(\xi)|<1$ for $\xi\neq 0$, Proposition~\ref{prop:no_normal_fixed_points} shows that this free set is empty within normal states. For completeness, we record here an alternative notion based on the image of $\Delta_g$.

A natural alternative is to call a state incoherent whenever it lies in the image of the dephasing map.

\begin{definition}\label{app:def:image_incoherent}
Given a dephasing channel $\Delta_g$ of the form Eqs.~\eqref{randomkicks}--\eqref{positionaction}, define
\begin{align}
  \mathcal I_{\mathrm{img}}^g \coloneqq \bigl\{\rho\ \text{state}:\ \exists\ \sigma\ \text{state such that}\ \rho=\Delta_g(\sigma)\bigr\}.
\end{align}
\end{definition}

Whenever $g(\xi)\neq 0$ almost everywhere, one may formally define the inverse kernel transform
\begin{align}\label{app_eq:inverseDelta}
  \bigl[\Delta_g^{-1}(\rho)\bigr](x,y)\coloneqq \frac{\rho(x,y)}{g(x-y)}.
\end{align}
This is only a \emph{partial inverse}: it is defined only for those operators for which the right-hand side corresponds to a valid operator. Note that this can fail in general, because the kernel $\bigl[\Delta_g^{-1}(\rho)\bigr](x,y)$ may lead to divergences when applied to $L^2(\mathbb{R})$ functions.

\begin{proposition}\label{app:prop:image_inverse}
Assume $g(\xi)\neq 0$ almost everywhere. A state $\rho$ belongs to $\mathcal I_{\mathrm{img}}^g$ if and only if $\Delta_g^{-1}(\rho)$ is well defined and is itself a state.
\end{proposition}

\begin{proof}
If $\rho=\Delta_g(\sigma)$ for some state $\sigma$, then by Eq.~\eqref{positionaction} one has $\Delta_g^{-1}(\rho)=\sigma$. Conversely, if $\Delta_g^{-1}(\rho)$ is a state, then $\rho=\Delta_g(\Delta_g^{-1}(\rho))$.
\end{proof}

Thus, inverse dephasing gives an equivalent characterization of image-based incoherence: a state is incoherent in this sense precisely when its inverse dephasing is again a physical state.

\begin{proposition}\label{app:prop:pure_not_image}
Assume $0<|g(\xi)|<1$ for almost every $\xi\neq 0$. Then no pure state belongs to $\mathcal I_{\mathrm{img}}^g$.
\end{proposition}

\begin{proof}
Let $\rho=\ketbra{\psi}{\psi}$ and suppose $\rho\in\mathcal I_{\mathrm{img}}^g$. Then $\widetilde\rho\coloneqq\Delta_g^{-1}(\rho)$ would be a state, with kernel
\begin{align}
  \widetilde\rho(x,y)=\frac{\psi(x)\psi(y)^*}{g(x-y)}.
\end{align}
Hence
\begin{align}
  \tr(\widetilde\rho^2) = \iint_{\mathbb{R}^2} \dd{x}\dd{y}\, \frac{|\psi(x)|^2|\psi(y)|^2}{|g(x-y)|^2}.
\end{align}
Since $|g(x-y)|<1$ away from the diagonal and every normalizable pure state has nonzero off-diagonal weight, this gives $\tr(\widetilde\rho^2)>1$, contradicting positivity and unit trace.
\end{proof}

This image-based notion is mathematically distinct from the fixed-point framework used in the main text. It is useful mainly as an alternative characterization of when a dephased state admits a physical preimage.

\section{Step-function dephasing at fixed spatial resolution}\label{app:step}

The CPTP channel $\Delta_g$ in the main text is physically motivated but generally non-idempotent. A complementary construction, useful for defining a sharp finite-resolution notion of incoherence, is obtained by replacing the smooth kernel $g$ with a hard cutoff in the relative coordinate.

For a fixed resolution parameter $\epsilon>0$, define
\begin{align}\label{app_eq:Iepsilon}
  \mathcal I_{\epsilon} \coloneqq \Bigl\{\rho\ \text{state}:\ \rho(x,y)=0\ \text{for }|x-y|>\epsilon\Bigr\}.
\end{align}
The associated step-function dephasing map is
\begin{align}\label{app_eq:DeltaTheta}
  \bigl[\Delta_\Theta^\epsilon(\rho)\bigr](x,y) = \Theta(\epsilon-|x-y|)\,\rho(x,y),
\end{align}
where $\Theta$ is the Heaviside step function.

\begin{proposition}\label{app:prop:step_basic}
The map $\Delta_\Theta^\epsilon$ is linear, trace preserving, and idempotent $\left(\Delta_\Theta^\epsilon\right)^2=\Delta_\Theta^\epsilon$. Moreover, $\rho\in\mathcal I_\epsilon \,\Leftrightarrow\, \Delta_\Theta^\epsilon(\rho)=\rho$.
\end{proposition}

\begin{proof}
Linearity is immediate. Since $\Theta(\epsilon-|x-x|)=1$, the diagonal is unchanged, so the trace is preserved. Idempotency follows from $\Theta^2=\Theta$. The fixed-point condition is equivalent to vanishing of the kernel for $|x-y|>\epsilon$.
\end{proof}

Thus $\Delta_\Theta^\epsilon$ plays the role of an exact projector onto a nontrivial finite-resolution free set of normal states. The price of this sharp construction is that $\Delta_\Theta^\epsilon$ is generally not completely positive. Indeed, for translation-invariant multiplicative kernel maps, complete positivity is equivalent to positive definiteness of the mask; for abelian translation groups, Bochner's theorem identifies such masks with characteristic functions of positive measures \cite{mckee2018positive}. The step mask fails this condition, since the Fourier transform of $\Theta(\epsilon-|\xi|)$ is proportional to $\frac{\sin(\epsilon p)}{p}$, which changes sign. Therefore $\Delta_\Theta^\epsilon$ should be understood as a mathematical projector defining the free set, Eq.~\eqref{app_eq:Iepsilon}, not as a physical channel.

This makes the contrast with the main text precise: the random-kick channel $\Delta_g$ is CPTP but non-idempotent, whereas $\Delta_\Theta^\epsilon$ is idempotent but generally not CP. 

Because $\Delta_\Theta^\epsilon$ is not CP, entropy-based expressions such as $S(\rho\Vert \Delta_\Theta^\epsilon(\rho))$ are not well defined for arbitrary states and are not used in the main text. Nevertheless, the step projector yields a simple Hilbert--Schmidt diagnostic of finite-resolution coherence.

\begin{definition}
    The finite-resolution Hilbert--Schmidt coherence diagnostic is
    \begin{align}\label{app_eq:C2Theta}
      C_{2}^{\epsilon}(\rho)
      &\coloneqq \iint \limits_{|x-y|>\epsilon}\dd{x}\,\dd{y}\,|\rho(x,y)|^2.
    \end{align}
\end{definition} 
This quantity measures the Hilbert--Schmidt weight of $\rho$ lying outside the strip $|x-y|\le\epsilon$. 

\begin{example}
For the normalized Gaussian wavefunction
\begin{align}
  \psi_\sigma(x) = \frac{1}{(2\pi\sigma^2)^{1/4}} \exp\!\left(-\frac{x^2}{4\sigma^2}\right),
\end{align}
Eq.~\eqref{app_eq:C2Theta} gives
\begin{align}\label{app_eq:gaussian_step_hs}
  C_{2}^{\epsilon}(\ketbra{\psi_\sigma}{\psi_\sigma}) = 1-\erf\!\left(\frac{\epsilon}{2\sigma}\right).
\end{align}
\end{example}

These closed forms explain the appeal of the step-function construction: it yields a sharp and analytically simple notion of finite-resolution coherence. Its role in the present work, however, is auxiliary. The main text is based instead on the physically implementable CPTP dephasing channel $\Delta_g$ and the coherence quantifiers induced directly by its action.

\end{document}